\documentclass[conference,compsoc]{IEEEtran}

\usepackage{graphicx}
\usepackage{subfigure}
\usepackage{bbm}
\usepackage{algorithm}
\usepackage{multirow}
\usepackage{algorithmic}
\usepackage{amsmath}
\usepackage{amsthm}
\newtheorem{myDef}{Definition}[section]
\newtheorem{myTheo}{Theorem}[section]
\newtheorem{myAss}{Assumption}[section]
\newtheorem{remark}{Remark}[section]

\ifCLASSOPTIONcompsoc
  \usepackage[nocompress]{cite}
\else
  \usepackage{cite}
\fi
\ifCLASSINFOpdf
\else
\fi

\begin{document}
%
% paper title
% Titles are generally capitalized except for words such as a, an, and, as,
% at, but, by, for, in, nor, of, on, or, the, to, and up, which are usually
% not capitalized unless they are the first or last word of the title.
% Linebreaks \\ can be used within to get better formatting as desired.
% Do not put math or special symbols in the title.
\title{An Input-Aware Mimic Defense Theory and its Practice}
\author{\IEEEauthorblockN{{Jiale Fu\IEEEauthorrefmark{2}, Yali Yuan\IEEEauthorrefmark{3}\IEEEauthorrefmark{1}, Jiajun He\IEEEauthorrefmark{2}, Sichu Liang\IEEEauthorrefmark{4}, Zhe Huang\IEEEauthorrefmark{3}, Hongyu Zhu\IEEEauthorrefmark{3}}
\IEEEauthorblockA{
% \IEEEauthorrefmark{*}University of G\"ottingen, Germany. Email:yaliyuan@seu.edu.cn\\
\IEEEauthorrefmark{2}Southeast University, School of Mathematics. Emails: \{213200900, 213203486\}@seu.edu.cn \\
\IEEEauthorrefmark{3}Southeast University, School of Cyber Science and Engineering. Email (\IEEEauthorrefmark{1}Corresponding): yaliyuan@seu.edu.cn\\
\IEEEauthorrefmark{3}Southeast University, School of Cyber Science and Engineering. Emails: 
 213200037, 213200421@seu.edu.cn\\
\IEEEauthorrefmark{4}Southeast University, School of Artificial Intelligence. Emails: 213200709@seu.edu.cn\\
%\IEEEauthorrefmark{4}Huawei Technologies Co., China. Email: kai.zheng@huawei.com
}}
\vspace{-5mm}
% \thanks{Yali Yuan\IEEEauthorrefmark{5} and Weijun Wang\IEEEauthorrefmark{5} have equal contributions in this work.}
\thanks{Corresponding author: Yali Yuan.}
%\thanks{Yali Yuan did this work at University of G\"ottingen, and now the author is an Associate Professor at Southeast University.}
% \thanks{Sripriya S. Adhatarao did this work at University of G\"ottingen, and now the author is working at Huawei Munich Research Center, Germany.}
}
% \thanks{corresponding author}
% author names and affiliations
% use a multiple column layouts for up to three different
% affiliations
% \author{\IEEEauthorblockN{Michael Shell}
% \IEEEauthorblockA{School of Electrical and\\Computer Engineering\\
% Georgia Institute of Technology\\
% Atlanta, Georgia 30332--0250\\
% Email: http://www.michaelshell.org/contact.html}
% \and
% \IEEEauthorblockN{Homer Simpson}
% \IEEEauthorblockA{Twentieth Century Fox\\
% Springfield, USA\\
% Email: homer@thesimpsons.com}
% \and
% \IEEEauthorblockN{James Kirk\\ and Montgomery Scott}
% \IEEEauthorblockA{Starfleet Academy\\
% San Francisco, California 96678-2391\\
% Telephone: (800) 555--1212\\
% Fax: (888) 555--1212}}

% make the title area
\maketitle
% As a general rule, do not put math, special symbols, or citations
% in the abstract
\begin{abstract}
The current security problems in cyberspace are characterized by strong and complex threats. Defenders face numerous problems such as lack of prior knowledge, various threats, and unknown vulnerabilities, which urgently need new fundamental theories to support. To address these issues, this article proposes a generic theoretical model for cyberspace defense and a new mimic defense framework, that is, Spatiotemporally heterogeneous, Input aware, and Dynamically updated Mimic Defense (SIDMD). We make the following contributions: (1) We first redefine vulnerabilities from the input space perspective to normalize the diverse cyberspace security problem. (2) We propose a novel unknown vulnerability discovery method and a dynamic scheduling strategy considering temporal and spatial dimensions without prior knowledge. Theoretical analysis and experimental results show that SIDMD has the best security performance in complex attack scenarios, and the probability of successful attacks is greatly reduced compared to the state-of-the-art.
\end{abstract}

% no keywords

% we design a dynamic update algorithm for executor vulnerabilities that can exploit historical attack information. (3) we first propose a dynamic scheduling strategy for executors based on spatiotemporal heterogeneity, which considers both temporal and spatial dimensions and significantly reduces the possibility that the system will expose vulnerabilities.

% For peer review papers, you can put extra information on the cover
% page as needed:
% \ifCLASSOPTIONpeerreview
% \begin{center} \bfseries EDICS Category: 3-BBND \end{center}
% \fi
%
% For peer review papers, this IEEEtran command inserts a page break and
% creates the second title. It will be ignored for other modes.
\IEEEpeerreviewmaketitle

\section{Introduction}

With the rapid development of networks, the popularity of the Internet has greatly facilitated our life. Meanwhile, network security issues receive increasing attention from the government and the academic world. However, existing network defense technologies, such as anti-virus, firewall, intrusion detection~\cite{1}, honeypot~\cite{2}, sandbox~\cite{3}, and intrusion tolerance~\cite{4}, can only handle attacks with known characteristics, which are essentially a “mending after a sheep is lost” methods.

In response to increasingly diverse cyberattacks, a large amount of research was conducted to propose effective defense technologies for specific types of attacks and application scenarios. For example, poisoning attacks, which affect the learning results by poisoning the training dataset, have been studied deeply for attack and defense against models, such as machine learning~\cite{11}-\cite{15}, federated learning~\cite{16}-\cite{18} and so on. Usually, mining the characteristics of the model itself by separating the attack and defense modules to provide targeted defense is the main research approach at present. In other words, if there is more than one type of attack against the model, its defense performance is likely to decline significantly in the long term. Therefore, to address the limitation that a point-to-point defense strategy can only resist specific attacks, it is necessary to introduce mimic defense techniques to defend against multiple attacks in cyberspace.

Moving Target Defense (MTD)~\cite{5} is an active defense strategy and enables point-to-area defense. MTD enhances the uncertainty of the system by constantly changing the attack surface of the system so that attackers cannot accurately detect system vulnerabilities. To detect unknown vulnerabilities and backdoors, Jiangxing Wu proposed Cyber Mimic Defense (CMD) technology ~\cite{6}. CMD is a dynamics, heterogeneity, and redundancy structure design based on the attributes of the system itself. Compared with MTD, CMD greatly increases the cost of attackers and gains higher robustness against multiple types of attacks. CMD theory has great potential in security performance, and many actual network scenarios are willing to pay more in exchange for reliable network security. In recent years, the mimic defense has been applied in Network Functions Virtualization (NFV)~\cite{7}, cloud service~\cite{8}, edge-computing~\cite{9}, distributed system~\cite{10}, etc., demonstrating a better defense effect.

Currently, research on mimic defense techniques mainly focuses on the impact of structure design, heterogeneous measurement, scheduling strategy, and adjudication mechanism on defense performance. The heterogeneity of mimic defense refers to the maximum heterogeneity of the hardware and software implementation of several executors capable of achieving the same function. The scheduling of mimic defense refers to the mimic scheduler dynamically updating several heterogeneous executors with a certain strategy and cleaning the expired executors offline. The scheduling strategy of mimic defense is an important foundation to ensure its dynamic and heterogeneous. A better scheduling strategy can make the mimic defense system expose fewer vulnerabilities and improve the security performance of the whole network.

However, in the real-world network security applications, the existing mimic defense has great shortcomings, mainly incorporating three points: (1) Although the simple definition of vulnerabilities brings great convenience for calculating the heterogeneity and similarity between executors, it is difficult to define the vulnerability for completely different executors in practice. (2) The heterogeneity of executors is static without considering the history attack knowledge. 
In practice, as the evaluation of attacks, new vulnerabilities may appear, which are hardly taken into account in existing studies. (3) The current scheduling strategy of mimic defense only considers the common vulnerabilities among several executors in the same space, which ignores that among adjacent schedules. Without considering the temporal heterogeneity, it is possible to have a situation where “each schedule has few common vulnerabilities, but these common vulnerabilities remain unchanged over varied schedules”.

To solve the aforementioned issues, in this paper, we construct a theoretical model of mimic defense from an input perspective for the first time and propose a new mimic defense framework named SIDMD. We make the following contributions: (1) for a better heterogeneity measure, we give an exact definition of vulnerability that is understandable and easy to use in practice from the perspective of the input space, and introduce a k-order heterogeneity calculation method. (2) we design a dynamic update algorithm for executor vulnerabilities that can exploit historical attack information. (3) we first propose a dynamic scheduling strategy for executors based on spatiotemporal heterogeneity, which considers both temporal and spatial dimensions and significantly reduces the possibility that the system will expose vulnerabilities. (4) we run numerous simulations and compare SIDMD to state-of-the-art to verify its effectiveness. Besides, we test SIDMD on real-world applications, for example, resisting attacks of machine learning classification over a real-world dataset, to further prove the proposed model. We elaborate on our contributions below.

On the heterogeneous measurement, we note that the input space of heterogeneous executors is the same regardless of the differences in structure, implementation, etc. First, we abstract the input space of the executor into a metric space and define the concept of “distance” between inputs. Secondly, based on the above definition, we give the vulnerability set of executors and define the index for similarity and heterogeneity calculation, so that the normalized comparison between different executors is possible. Thirdly, we design a dynamic update algorithm for executor vulnerabilities based on Fitness Proportionate Sharing (FPS) clustering, which can update the summary of the input stream based on historical attack information, reflecting the adaptability to the dynamic environment and evolutionary attack.

On the scheduling strategy, we propose a novel scheduling strategy considering spatiotemporal heterogeneity. In addition to considering spatial heterogeneity, this strategy can combine temporal information and select the best executor to be replaced at the current time based on the distribution of executors' vulnerabilities at the previous time. In this way, the potential harm can be minimized.
% Compared with spatial hetere rogeneity, which only considers the set of executor vulnerabilities at a single time, it can avoid the situation of few common mode vulnerabilities but a long existence time. It solves the limitation of the previous scheduling strategy that only considers spatial or temporal heterogeneity.

To evaluate SIDMD, first, we did numerous simulations. The results demonstrate that SIDMD performs well in discovering vulnerabilities, and the probability of being attacked by SIDMD is about 1/12 of the state-of-the-art. Besides, the total time of being attacked is about 1/14 of the state-of-the-art. The security performance of SIDMD with five online executors is similar to state-of-the-art with seven online executors. Second, we tested our model on a real-world application. The findings indicate that SIDMD is substantially more effective and robust than the state-of-the-art. Specifically, SIDMD consistently maintains a performance retention rate of 97\% when defending against various attacks, which is greater than Proda \cite{11} (96\% at most) and TRIM's rates \cite{12} (90\% at most), and SIDMD has good performance against different attacks, while the other two defense models have great differences in defense performance against different attacks. Additionally, the deployment of SIDMD does not reduce the models' classification accuracy.
%merely the deployment of SIDMD does not reduce the initial classification accuracy in the absence of assaults.

The rest of the paper is organized as follows. In Section~2, we provide the necessary background knowledge. In Section~3 we introduce our mimic defense theory. Section~4 provides the main framework of SIDMD. Subsequently, in Sections 5 and 6, we introduce the dynamic vulnerability update algorithm and the scheduling algorithm in detail, respectively. In Section~7, we analyze the complexity and cost of SIDMD and validate the defense performance of SIDMD through simulation and real datasets. Finally, we present related work in Section~8 and draw conclusions in Section~9.

% You must have at least 2 lines in the paragraph with the drop letter
% (should never be an issue)

\section{Preliminaries}

In this section, we introduce the main framework of traditional mimic defense and the necessary concepts and theorems.

\subsection{The Main Framework of Mimic Defense.}

\begin{figure}[H]
    \centering \includegraphics[width=8.3cm]{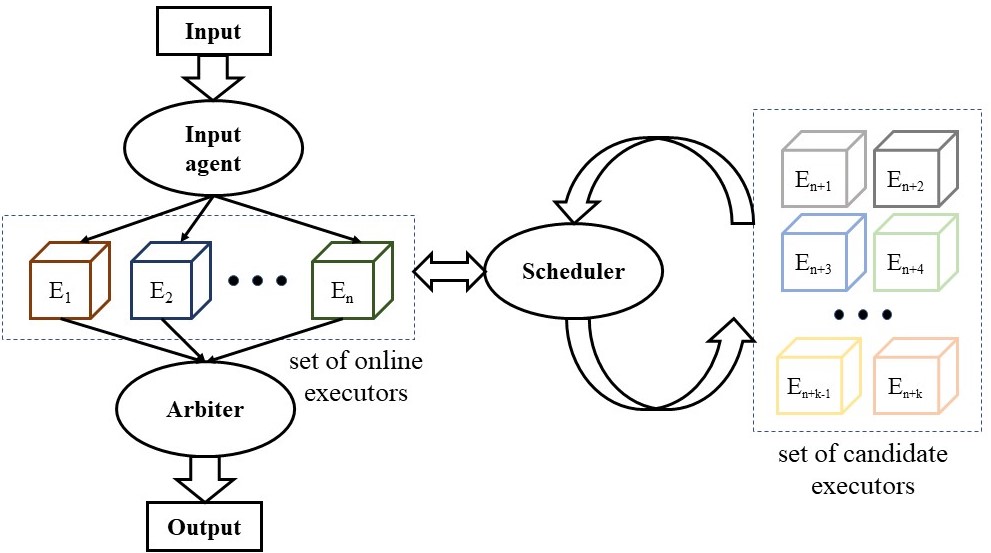}
    \caption{The main framework of mimic defense.}
    \label{fig:galaxy}
\end{figure}

The main framework of mimic defense is dynamic heterogeneous redundancy (DHR) \cite{6}. As shown in Figure~1, it is mainly composed of five parts, including input agents, heterogeneous executor pool, scheduler, online executor set, and arbiter. The details for each of them are discussed as follows.

\textbf{Input agents:} The input agent replicates and distributes the input to each executor of the online executor set.

\textbf{Heterogeneous executor pool:} The heterogeneous executor pool includes a variety of functionally equivalent and structurally different heterogeneous executors. The heterogeneous executor pool is an alternative pool in which executors can be scheduled into the online execution set and process input.

\textbf{Scheduler:} The scheduler selects a certain number of executors from the heterogeneous executor pool to be online with a certain scheduling algorithm.

\textbf{Online executor set:} An online executor set is a collection of heterogeneous executors running online, where each executor is independent and processes input data in parallel. Their outputs are aggregated to the arbiter.

\textbf{Arbiter:} The arbiter will aggregate the outputs of the online executors and generate voting results using a specific voting algorithm.

The heterogeneity of the executors in the heterogeneous executor pool ensures the heterogeneity of the DHR. The scheduler continuously schedules over time to ensure the dynamism of the DHR, and the online executor set ensures the redundancy of the DHR. The dynamism, heterogeneity, and redundancy of DHR make the system uncertain in time and space, making it difficult for attackers to exploit the vulnerabilities of the system so that the system can obtain endogenous security features.

\subsection{Metric Space}

We deploy in our work the metric space. so we recall here the definition of metric space.

\begin{myDef}
Let $A$ be a set, and let $d(\cdot,\cdot)$ be a binary real-valued function on $A$. We say $d(\cdot,\cdot)$ a metric on $A$, if $d(\cdot,\cdot)$ satisfies
\begin{enumerate}
\item $d(a,b)\geq 0, \forall a, b \in A$, and $d(a,b)=0 \Leftrightarrow a=b$;
\item $d(a,b)=d(b,a), \forall a, b \in A$;
\item $d(a,b) \leq d(a,c)+d(c,b), \forall a, b, c \in A$.
\end{enumerate}
\end{myDef}

Obviously, the metrics on a set are not unique, and we need to define different metrics according to different scenarios.

\subsection{Equivalence Relation and Set Partitioning}

We deploy in our work the equivalence relation and set partitioning related theory, so we
recall here the relevant theory.

\begin{myDef}
Let $A$ be a set, and let $R$ be a binary relation on $A$. We say $R$ a equivalence relation on $A$, if $R$ satisfies
\begin{enumerate}
\item $a R a $;
\item $a R b \Rightarrow b R a$;
\item $a R b, b R c \Rightarrow a R c $.
\end{enumerate}
where $a,b,c \in A$.
\end{myDef}

\begin{myTheo}
An equivalence relation on set A uniquely determines a partition of set A, which is there is a unique set of sets $A_1, A_2, \dots, A_n \subset A $, s.t.

\begin{enumerate}
\item $A=A_1\cup A_2\cup \dots \cup A_n$;
\item $A_i \cap A_j = \emptyset ,\  \forall i,j$;
\item $a R b \Leftrightarrow a,b \in A_i,\  \exists i$.
\end{enumerate}
\end{myTheo}

\section{Input-Aware Mimic Defense Theory}

In this section, we assume the input space of the executors to be a metric space. By defining the equivalence relation for the abnormal input set of the executor, we define the vulnerability of the executor, and then we considered the high-order similarity and heterogeneity of executors to better detect vulnerabilities in practice.

It is challenging to define the similarity and heterogeneity of the executors from the perspective of the executors themselves due to their differences. We note that their input space is always the same, regardless of how the two executors differ from one another. Therefore, from the viewpoint of the input space, executors can be contrasted in a normalized manner. On the basis of this idea, we offer the theoretical model below that describes the similarity and heterogeneity of executors.

\begin{myDef}
    Let $m$ denote an input to the executors and let $M=\left\{ m \right\}$ denote the set of all inputs, we also call $M$ the input space of executors.
\end{myDef}

In real-world applications, there are many different executors. Typically, all of our attention is focused on the output that the executor will produce for a given input and the accuracy of that output. Therefore, we can abstract the executors as a function with the input as an argument.

\begin{myDef}
Let function $f$ denote an executor, and let $F$ denote the set of all executors. Let $f(m)$ denote the output of executor $f$ for input $m$.
\end{myDef}

To state our results more conveniently, we introduce a perfect executor.

\begin{myDef}
Let function $f_0$ denote a perfect executor, which means that for any given input $f$, $f_0$ will always output correctly.
\end{myDef}

It can be seen that the result of $f_0(m)$ represents the output that the executor $f$ should produce, in other words, the correct output. But in reality, since each executor always has known or unknown vulnerabilities, $f_0$ is only an ideal model and has only theoretical significance. And we cannot know the value of $f_0(m)$. The perfect executor $f_0$ is introduced here only for the convenience of presentation.

\begin{myDef}
Let $m$ be an input and let $f$ be an executor. We say $m$ is a normal input of executor $f$, if $f(m) = f_0(m)$. We say $m$ an abnormal input of executor $f$, if $f(m) \ne f_0(m)$.
\end{myDef}

\begin{myDef}
We say an input $m$ is a known input if $m$ is what people find in practical applications. Let $m_s$ denote a known input and let $M_{s}=\bigcup \left\{ m_s \right\}$  denote the set of all known inputs.
\end{myDef}

Obviously, $M_s$ satisfies
\begin{equation}
M_s \subset M.
\end{equation}
And in practice, $M$ is usually an infinite set, and $M_s$ is usually a finite set.

\begin{myDef}
Let $f$ be an executor, and let $V(f)$ denote the set of all abnormal inputs of $f$ which is
\begin{equation}
\begin{split}
V\left( f \right) =\left\{ m\in M\,\,| f\left( m \right) \ne f_0(m) \right\}.
\end{split}
\end{equation}
\end{myDef}

\begin{myDef}
Let $f$ be an executor, and let $V_s(f)$ denote the set of all known abnormal inputs of $f$ which is
\begin{equation}
V_s(f)=V(f)\cap M_s.
\end{equation}
\end{myDef}

\begin{myAss}
Let $m_1,m_2$ be two abnormal inputs of executor $f$. If $m_1$ and $m_2$ are similar enough, that is, $d(m_1,m_2)$ is less than a given small positive number $\varepsilon_0$, then we can think that the abnormal input $m_1$ and $m_2$ are exploiting the same vulnerability of $f$.
\end{myAss}

In accordance with Assumption 3.1, we propose a definition of executors' vulnerability based on set partitioning and equivalence relations.

\begin{myDef}
Let $A$ be a subset of the input space $M$, let $a,b$ be two elements in A , let $d(\cdot,\cdot)$ be a metric on $M$, We say that $a$ and $b$ are related, denoted as $a \sim b$, if there are some inputs $m_1, m_2, \dots, m_n$, satisfy

\begin{enumerate}
\item $d\left( a,m_1 \right) <\varepsilon _0, d\left( m_n,b \right) <\varepsilon _0$;
\item $d\left( m_i,m_{i+1} \right) <\varepsilon _0, i=1,2,\cdots ,n-1$;
\item $m_1,m_2,\cdots ,m_n\in A$.
\end{enumerate}
where $\varepsilon_0$ is a small positive number.

\end{myDef}

Obviously, $"\sim"$ has the following properties.

\begin{enumerate}
\item $a \sim a $;
\item $a \sim b \Rightarrow b\sim a$;
\item $a \sim b, b \sim c \Rightarrow a \sim c $.
\end{enumerate}
where $a,b,c \in A$. Then, $"\sim"$ is an equivalence relation on set A.

\begin{myDef}
Let $f$ be a executor, let $V(f)$ be the set of all abnormal inputs of $f$, and let $"\sim"$ be the equivalence relation mentioned in Definition~3.8. According to Theorem~2.1, based on this equivalence relation, we can uniquely divide $V(f)$ into some disjoint subsets, and we say each subset is a vulnerability of the executor $f$.
\end{myDef}

After defining the vulnerabilities of the executors, we consider the similarity and heterogeneity of the executors. 

\begin{myDef}
Let $A$ be a subset of $M$. Let $P(A)$ represent the probability of drawing a sample from $M$ that happens to belong to $A$.
\end{myDef}

Obviously, the higher the probability that $f_1$ and $f_2$ are wrong at the same time, the more similar $f_1$ and $f_2$ are. So the similarity of $f_1$ and $f_2$ can be evaluated as the probability that $f_1$ and $f_2$ output the wrong result at the same time, i.e. $P(V(f_1)\cap V(f_2))$, the heterogeneity of $f_1$ and $f_2$ can be evaluated as 1-$P(V(f_1)\cap V(f_2))$. Next, we discuss the similarity and heterogeneity of executors in general.

\begin{myDef}
Let $f_1,f_2,\dots,f_n$ be $n$ executors, and the $k$-order similarity of these $n$ executors can be evaluated as: for a random input, the probability that $k$ of the $n$ executors make mistakes at the same time, and denoted as $P_k(f_1,f_2,\dots,f_n)$. $P_k(f_1,f_2,\dots,f_n)$ can be calculated as
\begin{equation}
\begin{split}
        P_k(f_1,f_2,\dots,f_n)=\sum_{i=1}^{C_n^k} P(A_i)-\sum_{1\leq i<j\leq C_n^k} P(A_i\cap A_j)  \\
         + \sum_{1\leq i<j<k \leq C_n^k} P(A_i\cap A_j \cap A_k)-\dots\\
         +(-1)^{C_n^k-1}P(A_1\cap \dots \cap A_{C_n^k}),
\end{split}
\end{equation}
where each $A_i$ represents an intersection of some $k$ executors' abnormal input sets.
\begin{equation}
A_i = \bigcap_{j=i_1,i_2,\dots,i_k} V(f_j),
\end{equation}
where $1 \leq j \leq n$ and $i_a \ne i_b, \forall a,b$.
\end{myDef}

The reason why we define the high-order similarity of multiple executors is that it can directly reflect the probability of arbitration failure. For example, if a voting-based arbitration DHR model deploys $2l+1$ executors, then the probability of the entire DHR model failing is equal to the probability of $l+1$ executors failing at the same time. Therefore, our method on the definition of high-order similarity can well evaluate the defense performance of DHR models. Naturally, k-order heterogeneity of executors can be evaluated as $1-P_k(f_1,f_2,\dots,f_n)$.

Obviously, to evaluate the high-order similarity and heterogeneity of multiple executors, the key is to calculate the value of the function $P(\cdot)$. However, in practical situations, it is difficult for us to obtain the abnormal input set $V(f)$ of the executor $f$ and the input space $M$. That is to say, it is almost impossible for us to get the value of $P(\cdot)$ by direct calculation. Therefore, we propose an approximate calculation of the function $P(\cdot)$ based on the known input set $M_s$.

\begin{myTheo}
    Assuming that the elements in $M_s$ are all extracted independently and identically from $M$, then for any subset $A$ of $M$, $\forall \varepsilon>0$ satisfy
    \begin{equation}
        \lim_{|M_s|\rightarrow \infty} Pr(|\frac{|M_s \cap A|}{|M_s|}-P(A)|<\varepsilon) = 1,
    \end{equation}
    where $Pr(\cdot)$ denotes the probability of something happening, and $|\cdot|$ denotes the number of elements in a set.
\end{myTheo}

\begin{proof}
This theorem is easy to be proved by Bernoulli's law of large numbers.
\end{proof}

According to Theorem~3.1, under the condition that the known inputs are independent and identically distributed, for any given subset $A$ of $M$, we can approximate the true probability $P(A)$ with $\frac{|A\cap M_s|}{|M_s|}$. We denote $\frac{|A\cap M_s|}{|M_s|}$ with $P'(A)$. Next, we introduce a method for finding an approximation $P'_k(f_1,f_2,\dots,f_n)$ of the k-order heterogeneity $P_k(f_1,f_2,\dots,f_n)$ based on the known input space $M_s$.

Let $f_1,f_2,\dots,f_n$ be some executors. A straightforward way to calculate the k-order similarity of executors $P'_k(f_1,f_2,\dots,f_n)$ is just like Equation~(4), and just replace $V(f_i)$ with $V_s(f_i)$. But such a definition is unreasonable. Even if two executors $f_1$ and $f_2$ are sufficiently similar, the intersection of $V_s(f_1)$ and $V_s(f_2)$ is probably going to be an empty set because the input space of the executors is frequently continuous and the known input space is only a finite number of points sampled from the input space. To address this, we provide a novel approach to computing intersections. We loosen the requirements for the intersection operation's elements, meaning that as long as two elements are sufficiently similar to one another, they can still be considered intersection operation elements.

\begin{myDef}
Let $A_1,A_2,\dots,A_n$ be some sets, and we say $a\in \bigcup_{i=1,\dots,n} A_i $ is a common element within $\varepsilon_0$ error tolerance of $A_1,A_2,\dots,A_n$, if
\begin{equation}
    \forall i=1,2,\dots,n, \ \exists a_i \in A_i,\ s.t. \ d(a,a_i)<\varepsilon_0.
\end{equation}

Let $(A_1 \cdot ,A_2 \cdot , \cdot \dots \cdot , \cdot A_n)_{\varepsilon_0}$ denote the set of all $a\in \bigcup_{i=1,\dots,n} A_i $ that satisfies the above condition, and we say $(A_1 \cdot ,A_2 \cdot , \cdot \dots \cdot , \cdot A_n)_{\varepsilon_0}$ the intersection of $A_1,A_2,\dots,A_n$ within $\varepsilon_0$ error tolerance.
\end{myDef}

The commutative law is easily verified to be satisfied by the set operation, hence the outcome of the operation has nothing to do with the set's order. The associative law, however, is not satisfied by this operation, thus it cannot be calculated by breaking it up into smaller parts and then synthesizing them.

In the case of no misunderstanding, we write $(A_1 \cdot ,A_2 \cdot , \cdot \dots \cdot , \cdot A_n)_{\varepsilon_0}$ as $A_1 \cdot ,A_2 \cdot , \cdot \dots \cdot , \cdot A_n$ for convenience. The known common abnormal input set can be defined as follows in accordance with Assumption 4.1.

\begin{myDef}
Let $f_1,f_2,\dots,f_n$ be some executors, and $V_s(f_1),V_s(f_2),\dots,V_s(f_n)$ are their known abnormal input set. We say $m\in \bigcup_{i=1,\dots,n} V_s(f_i)$ is a common abnormal input of $f_1,f_2,\dots,f_n$, if $m\in V_s(f_1) \cdot ,V_s(f_2) , \cdot \dots \cdot ,V_s(f_n)$, and we say $V_s(f_1) ,V_s(f_2) \cdot , \cdot \dots \cdot , V_s(f_n)$ the known common abnormal input set of executors $f_1,f_2,\dots,f_n$.
\end{myDef}

Based on this, we can evaluate k-order similarity and heterogeneity of executors by define $P'_k(f_1,f_2,\dots,f_n)$.

\begin{myDef}
Let $f_1,f_2.\dots,f_n$ be some executors, and $V_s(f_1),V_s(f_2),\dots,V_s(f_n)$ are their known abnormal input set. Then $P'_k(f_1,f_2,\dots,f_n)$ is calculated as
\begin{equation}
\begin{split}
        P'_k(f_1,f_2,\dots,f_n)=\sum_{i=1}^{C_n^k} P'(A_i)-\sum_{1\leq i<j\leq C_n^k} P'(A_i\cdot A_j)  \\
         + \sum_{1\leq i<j<k \leq C_n^k} P'(A_i\cdot A_j \cdot A_k)-\dots\\
         +(-1)^{{C_n^k}-1}P(A_1\cdot \dots \cdot A_{C_n^k}),
\end{split}
\end{equation}
where each $A_i$ represents an intersection of some $k$ executors' known abnormal input sets.
\begin{equation}
A_i =  V_s(f_{i_1})\cdot \dots \cdot V_s(f_{i_k}),
\end{equation}
where $1 \leq i_j \leq n$ and $i_a \ne i_b, \forall a,b$.
\end{myDef}

Naturally, the k-order heterogeneity of executors can be calculated as $1-P'_k(f_1,f_2,\dots,f_n)$.

\section{An Overview of SIDMD Framework}

In this section, we introduce the main framework of the SIDMD model, and a general algorithm description of the SIDMD framework is presented in Algorithm~1. Based on theoretical models in Section~3, we propose the framework of SIDMD: spatiotemporally heterogeneous, Input-aware, and dynamically updated mimic defense, and the framework of SIDMD is shown in Figure~2.

\begin{figure*}
    \centering \includegraphics[width=17.6cm]{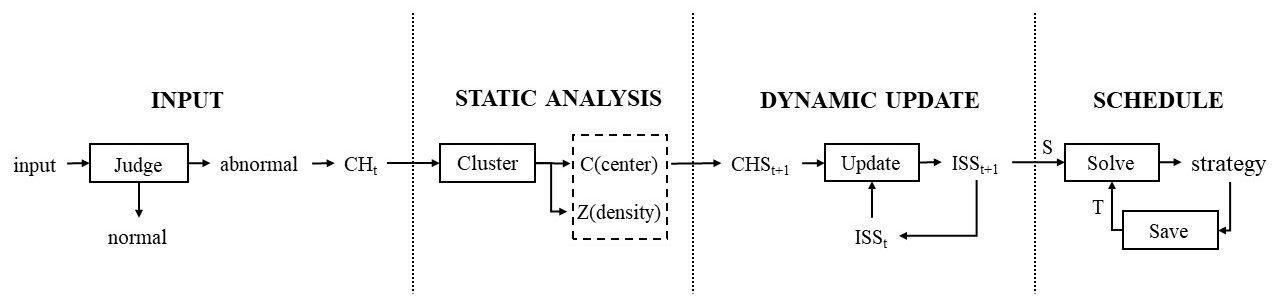}
    \caption{An overview of SID-DHR framework.}
    \label{fig:6}
\end{figure*}

First, the real-time input information will be judged as to whether it is an abnormal input for an executor. If so, it is saved and used to update the known abnormal input set of the executor. When an executor has saved enough abnormal inputs, these saved abnormal inputs are given as input chunks and summarized by the clustering algorithm. The summary of the input chunk $CHS$ will then be used to update the summary of the input stream $ISS$. Finally, When scheduling, the scheduling algorithm considers the factors of time and space and obtains the current optimal scheduling strategy according to the historical scheduling strategy and the summary of the input stream.

Let $IS$ be an input stream and let $CH_t$ be a chunk of input from $IS$ such that $IS=\bigcup^{\infty}_{t=1} \left\{ CH_t \right\}$, where $t$ refer to the time index of a input chunk. For each input chunk, $CH_t=\bigcup^{|CH_t|}_{i=1} \left\{ m_{t,i}|m_{t,i}\in M \right\}$, where $m_{t,i}$ denotes a input sample in $CH_t$. Let $CH_t^i$ to denote the abnormal input of the $i$th executor in $CH_t$. Let $IS^i$ represent the abnormal input stream of executor $i$, where $IS^i=\bigcup^{\infty}_{t=1} \left\{ CH_t^i \right\}$. Let $CHS_t^i = \left\{ C^{j}_{CH_t^i}, Z^{j}_{CH_t^i} \right\}$ be the summary of $CH_t^i$, where $C^{j}_{CH_t^i}$ represents the center of the $j$th cluster in $CH_t^i$ and the $Z^{j}_{CH_t^i}$ represents the number of elements in the cluster. Let $ISS_t^i$ be the summary of $IS^i$ up to time $t$, and the definitions of $C^{j}_{IS_t^i}, Z^{j}_{IS_t^i}$ are similar to the above. Let $SS_t$ represents the scheduling strategy at time $t$. Assume that there are a total of $N_p$ executors in the heterogeneous executor pool.

\begin{algorithm}[h]
	%\textsl{}\setstretch{1.8}
	\renewcommand{\algorithmicrequire}{\textbf{Input:}}
	\renewcommand{\algorithmicensure}{\textbf{Output:}}

	\caption{An overview of SIDMD framework}
	\label{alg3} 
	\begin{algorithmic}[1]
		\REQUIRE $IS$; $r$: the radius of the cluster; $r_0$: maximum distance to merge clusters.
		\ENSURE $SS_t$
\FOR{$t=1$ to $\infty$}

\IF{$ISS^i_t \ne \emptyset,\ \forall i$}
\STATE Use the scheduling algorithm to obtain the scheduling policy $SS_t$
\ELSE
\STATE Random schedule to get scheduling policy $SS_t$
\ENDIF

\FOR{$i=1$ to $N_p$}

\STATE Judge the abnormal input of the $i$th executor in $CH_t$ and get $CH_t^i$

\STATE Cluster abnormal inputs in $CH_t^i$ and get $CH_t^i$'s summary $CHS_t^i$

\IF{$ISS^i_t = \emptyset$}

\STATE $ISS_t^i=CH_t^i$

\ELSE

\STATE Use $CHS_t^i$ to update the $IS^i$'s summary up to time $t-1$ $ISS_{t-1}^i$ to obtain $ISS_{t}^i$
\ENDIF
\ENDFOR

\ENDFOR
	\end{algorithmic}
\end{algorithm}

\section{Dynamic Discovery of Vulnerabilities}

In this section, we describe how to dynamically discover vulnerabilities in real-world. The dynamic discovery of vulnerabilities can be divided into two parts: obtaining abnormal inputs for each executor in real-world and automatically updating vulnerabilities over time. Given that each executor tends to have an infinite number of abnormal inputs over time, the dynamic update algorithm we propose only needs to save a summary of each cluster, greatly reducing the algorithm's space complexity.

\subsection{Obtaining Abnormal Inputs}

It can be challenging to determine in advance whether an input is normal or aberrant in practice. As a result, we suggest a technique that may automatically assess whether the input is normal to an executor using SIDMD's arbitration mechanism.

For the input information, each executor will output a result, and these results will be combined to produce the arbitration mechanism's output. We compare the output of each executor with the output of the arbitration. It is deemed to be correctly outputting if the executor's result matches that of the arbiter; otherwise, it is deemed to be wrongly outputting. Algorithm~2 provides a full description of the process.

\begin{algorithm}[h]
	%\textsl{}\setstretch{1.8}
	\renewcommand{\algorithmicrequire}{\textbf{Input:}}
	\renewcommand{\algorithmicensure}{\textbf{Output:}}
	\caption{Obtaining abnormal inputs}
	\label{alg3} 
	\begin{algorithmic}[1]
		\REQUIRE $m$: input; $\left\{ f_1,f_2.\dots,f_n \right\}$: online executor set; Arbiter
		\ENSURE 0-1 vector $b$, whose $i$th component represents whether $m$ is an abnormal input of the $i$th executor
            \FOR{i=$1$ to $n$}
            \STATE $Result(i)=f_i(m)$
            \ENDFOR
		\STATE $ Result_0 \gets Arbiter(Result)$
            \FOR{i=$1$ to $n$}
		\IF{$Result(i) = Result_0$} 
		\STATE $b(i) \gets 0$ 
		\ELSE 
		\STATE $b(i) \gets 1$ 
		\ENDIF 
            \ENDFOR
	\end{algorithmic}
\end{algorithm}

\begin{figure}[h]
    \centering \includegraphics[width=8.3cm]{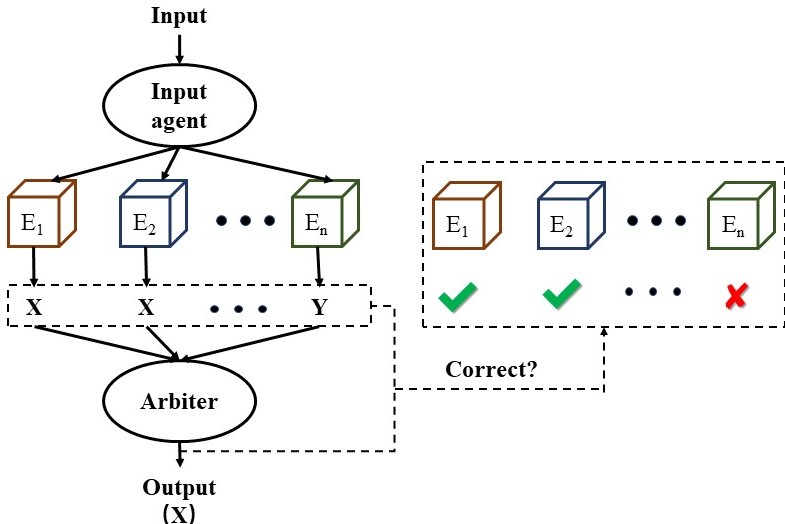}
    \caption{Schematic diagram of the process of obtaining abnormal input.}
    \label{fig:3}
\end{figure}

\begin{remark}
Although the outcome of the arbitration may not always be accurate, the probability that it is correct is substantially higher than the probability that a single executor will get the right answer. Therefore, it is feasible to judge the abnormal input of the executor in this way.
\end{remark}

\subsection{Dynamic Update Algorithms}

Each executor's vulnerabilities change over time in real-world applications. New vulnerabilities will inevitably arise as a result of the attacker's constant progress and the executor's wear and tear. Inspired by Yan et al. \cite{19}, we propose an algorithm to update the executor's vulnerability set in accordance with real-time aberrant input information. The detailed steps are listed below.

% We implement the necessary modifications to a streaming data clustering framework proposed by Yan et al. [19] to better fit our application situations. Next, we go into detail about how the algorithm actually works:

First, set parameters $r$ and $N$, for each $CH_t$, and randomly select $N$ candidate centers from $CH_t$ and evaluate their density as their fitness. With the Gaussian kernel type of density function used in \cite{19},
the fitness of each candidate center is defined as follows.
\begin{equation}
f(x_t^i)=\sum_{m=1}^N (e^{-\frac{d(x_t^i,x_t^m)}{\beta}})^\gamma, i \in [1,N],
\end{equation}
where $\beta$ denotes the variance of data, and $\gamma$ denotes the scaling parameter that approximates the shape of clusters. $N$ is the number of candidate centers. And $d(x_t^i,x_t^m)$ represents the metric between $x_t^i$ and $x_t^m$.

\begin{remark}
    The reason why $N$ candidate cluster centers are selected from $CH_t$, rather than the entire $CH_t$ as a candidate cluster center, is to prevent the computational complexity from being too high. Therefore, when the number of elements in the input chunk is relatively small, the entire a can also be used as a candidate cluster center.
\end{remark}

Next, we choose the input with the highest density value as the center of the first cluster. Then, the input distances less than $r$ from the center of the cluster are taken as the elements in the first cluster, and fitness proportionate are shared among them. The new fitness (density) is evaluated as follows.
\begin{equation}
f(x_t^{ij})_{new}=
\frac{f(x_t^{ij})_{old}}{\sum_{j=1}^{n_i} f(x_t^{ij})_{old}},
\end{equation}
where $x_t^{ij}$ denotes the $j$th element in the $i$th cluster at time $t$. $f(\cdot)_{new}$ and $f(\cdot)_{old}$ denote the new fitness and old fitness.

Then, we choose the input with the highest fitness at this time as the center of the second cluster, and so on, until all candidate cluster centers belong to a certain cluster. The pseudo-code for this process is provided in Algorithm~3.

\begin{algorithm}
	%\textsl{}\setstretch{1.8}
	\renewcommand{\algorithmicrequire}{\textbf{Input:}}
	\renewcommand{\algorithmicensure}{\textbf{Output:}}
	\caption{Obtain the summary of abnormal inputs} 
	\label{alg3} 
	\begin{algorithmic}[1]
		\REQUIRE $CH_t$: input chunk; $N$: the size of the candidate center set; $r$: the radius of the cluster; $C_t^i$: the center of the $i$th cluster of $CH_t$; $Z_t^i$: the density of $CH_t$.
		\ENSURE $CHS_t$
  
            \STATE randomly select $N$ samples in $CH_t$ as cluster center candidates denoted as $CH_t'$
            
            \FOR{i=$1$ to $N$}
            \STATE $Z_t^i \gets f(CH_t'(i))$
            \ENDFOR
            \STATE $i \gets 1$
            \WHILE{exists an element in $CH_t'$ that doesn't belong to any cluster}
            \STATE Select the input with the highest fitness in $CH_t'$ as the $i$th cluster center, denotes as $m_i$
            \STATE Take the input whose metric to the $i$th cluster center is less than $r$ as the elements in the $i$th cluster
            \STATE Do fitness proportionate sharing to them
            \STATE $i \gets i+1$
            \STATE $C_t^i \gets m_i$
            \STATE $Z_t^i \gets$ the number of elements in the $i$th cluster
            \ENDWHILE
		\STATE $CHS_t \gets \left\{ C_t, Z_t \right\}$
	\end{algorithmic}
\end{algorithm}

Next, we introduce how to dynamically update the summary of an executor. Assuming that the summary at time $t+1$ is $CHS_{t+1}$, and up to time $t$, the summary of the input stream is $ISS_t$. We set the parameter $r_0$, and assume that when the distance between the two cluster centers is less than $r_0$, they should be the same cluster, that is, the difference between the two cluster centers is considered to be caused by random factors. The way to update $ISS_t$ is as follows.

For each cluster in $CHS_{t+1}$, denoted as $C_{CH_{t+1}}^{i_1}$, if there are some clusters in $ISS_t$, denoted as $C_{IS_t}^{i_2}$, satisfying $d(C_{CH_{t+1}}^{i_1},C_{IS_t}^{i_2})<r_0$, they will be merged into one cluster. Suppose the merged cluster is $C_{IS_{t+1}}^{i_3}$ and $Z_{IS_{t+1}}^{i_3}$ is its density, then the specific update formula is as follows.
\begin{equation}
\left\{
\begin{array}{cl}
C_{IS_{t+1}}^{i_3}=& C_{CH_{t+1}}^{i_1}\\
Z_{IS_{t+1}}^{i_3}=& Z_{CH_{t+1}}^{i_1}+\sum_{j=1}^{|ISS_t|} C_{IS_t^i}^{j} \mathbbm{I}(d<r_0)\\
\end{array} \right.,
\end{equation}
where $\mathbbm{I}(\cdot)$ denotes indicative function and $d$ denotes $d(C_{CH_{t+1}}^{j},C_{IS_t}^{j})$. It takes value 1 when the event happens and value 0 when the event does not happen. The remaining clusters either belong only to $CHS_{t+1}$ or only to $IS_{t}$, we add it to the new cluster $ISS_{t+1}$. The pseudo-code for this process is provided in Algorithm~4.

\begin{algorithm}
	%\textsl{}\setstretch{1.8}
	\renewcommand{\algorithmicrequire}{\textbf{Input:}}
	\renewcommand{\algorithmicensure}{\textbf{Output:}}
	\caption{Update the summary} 
	\label{alg3} 
	\begin{algorithmic}[1]
		\REQUIRE $ISS_t$: the summary of input stream up to time $t$; $CHS_{t+1}$: the summary of input chunk at time $t+1$; $r_0$: maximum radius for cluster merging.
		\ENSURE $ISS_{t+1}$
            \FOR{i=$1$ to $|CH_{t+1}|$}
            \STATE $C_{IS_{t+1}}^i \gets C_{CH_{t+1}}^i$
            \STATE $Z_{IS_{t+1}}^i \gets Z_{CH_{t+1}}^i+\sum_{j=1}^{|IS_{t}|} Z_{CH_{t+1}}^j \mathbbm{I}(d<r_0) $
            \ENDFOR
            \STATE Add clusters in $ISS_t$ that have not been calculated to $ISS_{t+1}$
	\end{algorithmic}
\end{algorithm}

In order to ensure the vulnerability update capability of SIDMD, we can regularly normalize the density of all clusters, and delete clusters with small density. This not only adds a forgetting mechanism to SIDMD to ensure its update capability but also prevents a cluster from being too large to save. Based on this, the calculation of k-order similarity and heterogeneity between executors only needs to turn the original count into a weighted sum, where the weight is the density of the cluster center. The value $P'(A)$ is calculated as follows.
\begin{equation}
P'(A)=\frac{\sum_{i} Z_i \mathbbm{I}(C_i \in A)} {\sum_{i} Z_i }.
\end{equation}
Based on this, we can derive k-order similarity and heterogeneity calculation methods.

\section{Scheduling Strategy based on Spatio-Temporal Heterogeneity}

In this section, we comprehensively consider the spatial heterogeneity between executors and the temporal heterogeneity between online executor sets and propose an index to evaluate potential harm. Based on this index, we suggest a scheduling strategy that can reduce potential harm. The scheduling strategy also maintains uncertainty, guaranteeing that our model is both learnable and nondeterministic.

\textbf{Spatial heterogeneity:} For each scheduling cycle, each executor in the online executor set has its vulnerabilities. At this time, the security of the system is determined by executors' common vulnerabilities, which are focused on by earlier mimic protection models.

\textbf{Temporal heterogeneity:} Spatial heterogeneity does not consider the duration of common vulnerabilities; it only considers the extent of common vulnerabilities all at once. If the common vulnerability of a mimic defense system is small but its common vulnerability remains the same after scheduling, then even if the probability of an attacker's successful attack is small, once he succeeds, it will pose a long-term threat to the system and cause great losses. Furthermore, in practical scenarios, attackers are often memorable, which means that they are more likely to employ attack strategies that have worked in the past. Therefore, we also need to minimize the intersection of vulnerabilities exposed by adjacent scheduling cycles.

We use the historical scheduling information to guide the scheduling of the next scheduling cycle.

\begin{myAss}
    The attacker has the same probability of taking an attack action at each moment. Let $p_{att}$ denote it.
\end{myAss}

\begin{myDef}
    Let $O(t)$ represent the set of common abnormal inputs of the online executor set at time $t$.
    \begin{equation}
        O(t)=V(f_{t,1})\cdot V(f_{t,2})\cdot \dots \cdot V(f_{t,n}),
    \end{equation}
    where $V(f_{t,i})$ denotes the abnormal input set of the $i$th executor at time $t$.
\end{myDef}

Let $p_{t}$ represent the probability of the attack being successful if the attacker attacks at time $t$, and $p_{t}$ can be calculated as
\begin{equation}
    p_t = P(O(t)).
\end{equation}
From the previous analysis, $p_t$ can be estimated by $P'(O(t))$. Let $p_{t\rightarrow t+k}$ represent the probability that the attack at time $t$ is still effective at time $t+k$. If the attack launched at time $t$ is still valid at time $t+k$, it must be valid in every scheduling cycle between $t$ and $t+k$, otherwise, it will be discovered by the defense model. So $p_{t\rightarrow t+k}$ can be calculated as
\begin{equation}
    E_{t\rightarrow t+k}=\frac{|O(t)\cdot \dots\cdot O(t+k)|}{|O(t)|}.
\end{equation}
We hope that through scheduling, the attacks are as few as possible in the $t$ scheduling period. The attack in the scheduling period $t$ is either a successful attack launched by the attacker at time $t$ or a successful attack launched before time $t$ and still valid at time $t$. The probability that the attack launched at time $t$ succeeds and enters the system is $p_t$, and the probability that an attack launched at time $t-k$ succeeds and is still effective at time $t$ is $p_{t-k} \cdot p_{t-k \rightarrow t}$. Let $E_t$ denote the expectation of attacks in the scheduling period $t$. For a given scheduling policy, The calculation method of $E_{t}$ is as follows.
\begin{equation}
    E_t=p_t+p_{t-1\rightarrow t} p_{t-1}  +p_{t-2\rightarrow t} p_{t-2} +\dots.
\end{equation}
Let $R(n)$ denote the remainder of $E_t$.
\begin{equation}
    R(n)=p_{t-k\rightarrow t} p_{t-n}  +p_{t-(n+1)\rightarrow t} p_{t-(n+1)} +\dots.
\end{equation}
Assuming $E_t$ is bounded, then when $n$ is large enough, $R(n)$ can be arbitrarily small. Therefore, we can discard the remainder and compute only the top $k$ terms of the $E_t$.

Equation~(17) shows the expectation of attacks in the scheduling period $t$, but in practical application scenarios, we prefer to minimize the loss caused by the attack. Next, we modify Equation~(17) to estimate the expectation of the loss caused by the attack.

\begin{myAss}
    The destructive capabilities of the attackers are all the same. In other words, the loss caused by the attacker to the system through the attack is only related to the duration of the attack.
\end{myAss}

Let $h(T)$ denote the loss to the system caused by the attacker's continuous attack for $T$ duration. Then, the expectation of the loss caused by the attack can be estimated as
\begin{equation}
\begin{split}
        E(Harm)_t = p_t h(1) +p_{t-1\rightarrow t} p_{t-1} (h(2)-h(1))  +\\ p_{t-2\rightarrow t} p_{t-2} (h(3)-h(2)) +\dots.
\end{split}
\end{equation}
Therefore, the choice of scheduling strategy can be transformed into an optimization problem: Minimize $E(Harm)_t$ by choosing the online executors set. We then solve this optimization problem using algorithms such as Monte Carlo simulations or heuristic methods that guarantee some randomness.

\section{Analysis}

\subsection{Complexity Analysis}

We discuss the computational complexity of our proposed algorithm separately.

\textbf{Dynamic Update Algorithm:} To perform a clustering algorithm based on fitness proportionate sharing, it is first necessary to calculate $\frac{N (2n-N-1)}{2} $ metrics to obtain the fitness value of each candidate cluster center, where $n$ represents the total number of the input chunk and $N$ represents the number of candidate cluster centers. Then, at most $N$ computations are required for fitness proportional sharing. Therefore, the total number of calculations of the clustering algorithm based on fitness proportionate sharing is $\frac{N (2n-N+1)}{2} $.

Next, we discuss the computational complexity of the process of merging clusters. We assume that an executor saves at most $N_c$ clusters. The cluster merging process requires us to compute at most $N_c\times N_c$ metrics and at most $N_c$ cluster density updates. Therefore, the overall complexity of the clustering algorithm based on fitness proportionate sharing is $N_c^2+N_c$.

\textbf{Scheduling Algorithm:} Calculating the $p_t$ requires $\frac{n(n-1)}{2} N_c^2$ metric calculation. Calculating $p_{t\rightarrow t+k}$ requires at most $k n^2 N_c^2$ metric calculations. Suppose we utilize the scheduling information of the previous $L_{ST}$ times, then calculating $E_t$ requires at most $\frac{1}{2} n N_c^2 (L_{ST}+1)(n L_{ST}+n-1)$ metric calculations. Assuming that $E_t$ needs to be calculated $K$ times in the heuristic algorithm, the total computational complexity of the scheduling algorithm is $ \frac{1}{2} K n N_c^2 (L_{ST}+1)(n L_{ST}+n-1)$.

\subsection{Cost analysis}

The cost of our method can be mainly divided into the following four parts: the cost of heterogeneity, the cost of dynamic, the cost of redundancy, and the cost of cleaning and refactoring the executors, which we discuss in detail below.

\textbf{The cost of heterogeneity:} When deploying a mimic defense system, a large number of heterogeneous executors are required to ensure the security of the system. We assume that the cost of constructing each heterogeneous executor is the same, and denoted by $cost_f$. Then, the cost of heterogeneity can be evaluated as:
\begin{equation}
    Cost_H=cost_f \times N_f,
\end{equation}
where $N_f$ represents the number of executors in a heterogeneous executor pool.

\textbf{The cost of dynamic:} When performing mimic scheduling, replacing the executor requires a certain cost. This cost is related to the size of the online executor set, and the replacement of executors will have a certain impact on the normal work of the system, such as causing additional delays to the system. The cost of replacing the executor is linearly related to the number of the size of online executor set, and the additional loss caused by scheduling is often a constant. Therefore, the cost of executing a schedule can be evaluated as
\begin{equation}
    Cost_{D_T}=cost_{D_1}+n \times cost_{D_2},
\end{equation}
where $cost_{D_1}$ represents the cost of replacing an executor, $cost_{D_2}$ represents the loss from scheduling, $n$ represents the size of online executor set. Assuming that the scheduling period is $T$, the scheduling cost per unit time is
\begin{equation}
    Cost_{D}=\frac{cost_{D_1}+n \times cost_{D_2}}{T}.
\end{equation}

\textbf{The cost of redundancy:} A reliable mimic defense system needs to run multiple executors with the same function at the same time, which brings an additional cost to the whole system. In unit time, the cost of redundancy can be evaluated as
\begin{equation}
    Cost_R=cost_{R_0}\times n,
\end{equation}
where $cost_{R_0}$ represents the cost per unit of time to run an executor.

\textbf{The cost of cleaning and refactoring the executors:} In practice, we often need to clean and refactor some fragile executors, which will bring a certain cost. We assume that the cost of cleaning and refactoring each executor is the same, and let $cost_{C}$ denote it. $Cost_{C}$ is linearly related to the number of times the system cleans and refactors the executors.

$Cost_H$ is a one-time cost, and $Cost_D$, $Cost_R$ are costs per unit of time, which means that the cost of dynamic and redundancy will continue to increase over time, and the cost of refactoring and cleaning the executor depends on the actual situation. Then, the total cost $Cost_T$ can be expressed as
\begin{equation}
    Cost_T=Cost_H + t(Cost_{R_0}+Cost_{D_0}) + Cost_C \times N_C,
\end{equation}
where $t$ represents the length of time the SIDMD runs. $N_C$ stands for the number of times cleaning and operations were performed.

\subsection{Simulation}
We did random simulations to verify the reliability of SIDMD and compare SIDMD with state-of-the-art like DHR \cite{6}. We assume that the input space is the region $[-50,50]\times [-50,50]$ on a two-dimensional Euclidean space and randomly generate some circular areas in this area as the vulnerabilities of each executor. We randomly generate some points within this region as inputs.

We first verify the effectiveness of the dynamic update algorithm of the SIDMD model. We randomly generate vulnerabilities for each executor and then compare them with the clusters discovered and saved by the SIDMD model. The result is shown in Figure~4. We can see from Figure~4 that the vulnerabilities discovered by the SIDMD model are very similar to the real ones.

\begin{figure*}[t]
\centering
\subfigure[real vulnerabilities]{
\label{Fig.sub.1}
\includegraphics[width=8.3cm]{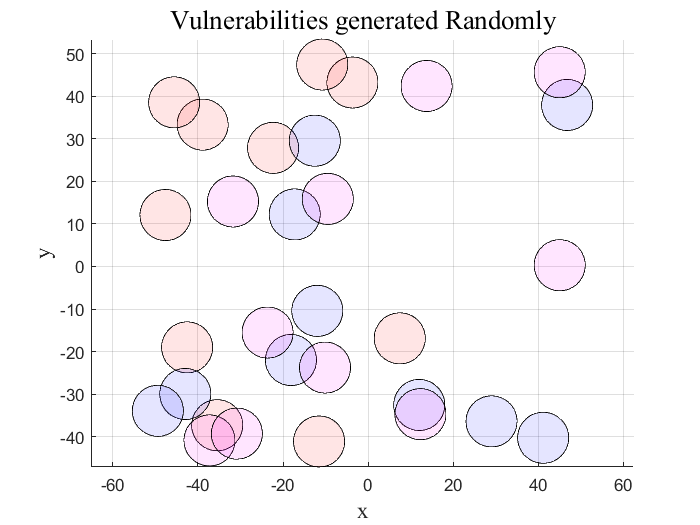}}\subfigure[clusters discovered and saved by the SIDMD]{
\label{Fig.sub.2}
\includegraphics[width=8.3cm]{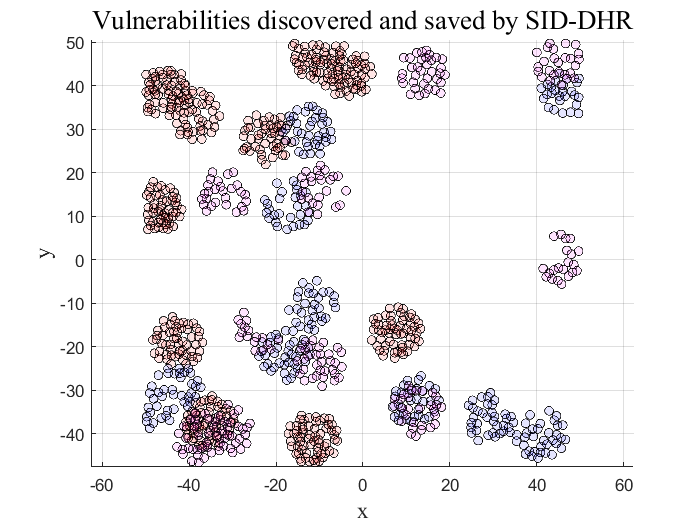}}
\caption{Comparison between true vulnerabilities and discovered vulnerabilities.}
\label{1}
\end{figure*}

Next, we verify that the performance of the DHR model is greatly improved after adding the dynamic update algorithm and the scheduling algorithm based on spatiotemporal heterogeneity. The settings of the relevant parameters and results can be seen in Table~1 and Table~2. (We say the DHR model added the dynamic update algorithm and the heterogeneity calculation algorithm as IDMD). We use a Monte Carlo simulation algorithm to solve the optimal scheduling policy.
\begin{table}[H]
\centering
\caption{Main parameters and values.}
\begin{tabular}{lr} 
\hline
Main parameters                              & \multicolumn{1}{l}{Value}  \\ 
\hline
Total number of executors                    & 50                         \\
Number of vulnerabilities of each
  executor & 10                         \\
Radius of the vulnerability                                         & 6                          \\
Total number of inputs              & 100000                      \\
Size of the online executor set              & 5                          \\
Scheduling period                            & 5000                       \\
$L_{ST}$                            & 2                       \\
$r$                            & 2                       \\
$r_0$                            & 4                       \\
\hline
\end{tabular}
\end{table}

\begin{table}[H]
\centering
\caption{Comparison of DHR, IDMD, and SIDMD.}
\begin{tabular}{l|lll} 
\hline
Index & DHR \cite{6}       & IDMD     & SIDMD     \\ 
\hline
N     & 590       & 58       & 47        \\
P     & 0.98\%    & 0.10\%   & 0.08\%    \\
T     & 1671226   & 156926   & 119430    \\
n     & 36        & 2        & 0         \\
P2    & 0.0610169 & 0.034483 & 0         \\
ET    & 2832.5864 & 2705.621 & 2541.064  \\
\hline
\end{tabular}
\end{table}

In Table~2, In Table~2, N, P, T, P2, and ET, represent the number of successful attacks, the probability of the attacker successfully attacking, the total time of being attacked, the number of times the attacker can still successfully attack after scheduling, the probability that the attacker can still attack successfully after scheduling, and the average attack time, respectively. From Table~2, we can see that the defense performance of IDMD and SIDMD is significantly better than DHR, and IDMD and SIDMD perform similarly in the number of attacks and the probability of being attacked. However, since the IDMD model only focuses on minimizing the common vulnerabilities at the same time, and thus ignores the relationship between the common vulnerabilities in adjacent scheduling cycles, the total duration of the attack is significantly longer than that of SIDMD. The average attack time of the IDMD model is also significantly higher than the SIDMD. Therefore, SIDMD can not only reduce the probability of vulnerability exposure but also reduce the duration of the same vulnerability exposure, which is the best among the three models.

\begin{remark}
    Since we assume that we do not know the vulnerabilities of each executor at the beginning, the IDMD and SIDMD model also uses a randomly selected scheduling method at the beginning. Until the vulnerability of each executor is updated to a certain level before starting to execute our scheduling strategy. Therefore, the number of attacks here is not counted from the beginning, but from the execution of our scheduling policy.
\end{remark}

Finally, we study the difference in cost between SIDMD and DHR when their defensive performance is similar. After many experiments, we found that the security performance of SIDMD with 5 online executors is similar to that of DHR with 7 online executors. Our parameter settings and results are shown in Table~3 and Table~4, respectively.

\begin{table}[H]
\centering
\caption{Parameter settings of SIDMD with 5 online executors and DHR with 7 online executors.}
\begin{tabular}{l|cc} 
\hline
\multirow{2}{*}{Main
  parameters}                                                   & \multicolumn{2}{c}{Value}  \\
                                                                                     & DHR             & SIDMD   \\ 
\hline
Total number of executors                                                            & \multicolumn{2}{c}{25}     \\
\begin{tabular}[c]{@{}l@{}}Number of vulnerabilities of each executor\end{tabular} & \multicolumn{2}{c}{10}     \\
Radius of the vulnerability                                                          & \multicolumn{2}{c}{6}      \\
Total number of inputs                                                               & \multicolumn{2}{c}{50000}  \\
\begin{tabular}[c]{@{}l@{}}Size of the online executor set\end{tabular}            & 7                & 5       \\
scheduling period                                                                    & 5000             & 5000    \\
\$L\_\{ST\}\$                                                                        & \textbackslash{} & 2       \\
\$r\$                                                                                & \textbackslash{} & 2       \\
\$r\_0\$                                                                             & \textbackslash{} & 4       \\
\hline
\end{tabular}
\end{table}

\begin{table*}[t]
\centering
\caption{The security performance of SIDMD with 5 online executors and DHR with 7 online executors.}
\begin{tabular}{c|cc|cc|cc|cc|cc|cc} 
\hline
\multirow{2}{*}{Num} & \multicolumn{2}{c|}{N} & \multicolumn{2}{c|}{P} & \multicolumn{2}{c|}{T} & \multicolumn{2}{c|}{n} & \multicolumn{2}{c|}{P2} & \multicolumn{2}{c}{ET}  \\
                     & DHR    & SID           & DHR    & SID           & DHR       & SID        & DHR   & SID            & DHR  & SID              & DHR     & SID           \\ 
\hline
1                    & 404    & 331           & 0.81\% & 0.66\%        & 1273480   & 828014     & 57    & 5              & 0.14 & 0.02             & 3152.18 & 2501.55       \\
2                    & 291    & 157           & 0.58\% & 0.31\%        & 906417    & 422970     & 33    & 0              & 0.11 & 0.00             & 3114.84 & 2694.08       \\
3                    & 340    & 320           & 0.68\% & 0.64\%        & 889117    & 842954     & 2     & 8              & 0.01 & 0.03             & 2615.05 & 2634.23       \\
4                    & 315    & 255           & 0.63\% & 0.51\%        & 793453    & 645269     & 2     & 0              & 0.01 & 0.00             & 2518.90 & 2530.47       \\
5                    & 214    & 292           & 0.43\% & 0.58\%        & 582114    & 797379     & 18    & 12             & 0.08 & 0.04             & 2720.16 & 2730.75       \\
6                    & 215    & 307           & 0.43\% & 0.61\%        & 533119    & 821392     & 2     & 6              & 0.01 & 0.02             & 2479.62 & 2675.54       \\
7                    & 247    & 331           & 0.49\% & 0.66\%        & 827266    & 870857     & 26    & 12             & 0.11 & 0.04             & 3349.26 & 2630.99       \\
8                    & 129    & 331           & 0.26\% & 0.66\%        & 331022    & 803925     & 0     & 0              & 0.00 & 0.00             & 2566.06 & 2428.78       \\
9                    & 210    & 237           & 0.42\% & 0.47\%        & 519501    & 589094     & 4     & 2              & 0.02 & 0.01             & 2473.81 & 2485.63       \\
10                   & 318    & 209           & 0.64\% & 0.42\%        & 867537    & 528037     & 14    & 0              & 0.04 & 0.00             & 2728.10 & 2526.49       \\
Ave                  & 268.30 & 277.00        & 0.01   & 0.01          & 752302.60 & 714989.10  & 15.80 & 4.50           & 0.05 & 0.01             & 2771.80 & 2583.85       \\
\hline
\end{tabular}
\end{table*}

\subsection{Validation on real dataset}

Darknet traffic classification is important for classifying real-time applications. Analyzing dark web traffic can help monitor malware ahead of time and detect malicious activity after an attack has broken out. Our experimental data is selected from the darknet dataset: https://www.unb.ca/cic/datasets/darknet2020.html. There are about 260,000 pieces of data and 78 valid features.

We verify the effectiveness of SIDMD at a specific moment in the poisoning attack scenario. We choose one or more of StatP \cite{11}, Nopt \cite{12}, alfa \cite{13} as the attack, and TRIM \cite{11}, Proda \cite{12}, K\_LID\_SVM \cite{13} as the defense respectively. We select 2000 data sets as datasets with a class balance, among which, 1800 are used as training sets (for \cite{13} we chose 1600 training sets, 200 validation sets), and 200 as validation sets. The poisoning rate is set to 20\%, and the mixed attack is set to 10\% each, and they attack without interfering with each other. 

The linear model without any regularization is selected as the linear model of StatP \cite{11} and Proda \cite{12}, the initialization of the attack selection is InvFlip, the optimization argument is $x$ (instead of $x$ and $y$), and the model is a white-box model. For defense, the maximum number of iterations $\gamma$ is set to 30. Since the model used above are regression-oriented, we convert it to binary classification by setting a threshold.

SVM [13] is a soft-interval Gaussian kernel SVM, the model parameter selection is consistent with the original experimental code, C = 0.39685026299204973 GAMMA = 0.7937005259840995. The attack adopts an adaptive poisoning attack, and the defense is the default defense.

SIDMD selects 5 classifiers as its executive body, including linear model, SVM, KNN (cluster classifier), random forest (tree structure classifier), and Naive Bayes (probabilistic model classifier). Without loss of generality, both the linear classifier and the SVM classifier in DHR are the same as above.

The experimental results we obtained are shown in Table~5. Given the different baselines of different models, the accuracy cannot be directly compared. We compare the accuracies before and after the attack to obtain the performance retention rate, as shown in Table~6. From this result, we can conclude that the defense performance of the SIDMD model is effective and robust to different attacks.

\textbf{Effectiveness:} From Table~6, SIMDM has good defensive performance against various attacks, including against different individual attacks and mixed attacks.

\textbf{Robustness:} For different types of attacks, the performance of Proda and TRIM is not stable enough: Proda resists alfa and StatP mixed attacks, the performance retention rate is 0.9595, but when resisting StatP attacks, the performance retention rate is 0.9133; when TRIM resists alfa attacks, the performance The retention rate is 0.896, but the performance retention rate is only 0.8092 when defending against StatP attacks. But SIDMD can show a good and stable defense effect against different kinds of attacks.

\begin{table}[H]
\centering
\caption{Multiple Attacks and Multiple Defense Test Results: Accuracy.}
\begin{tabular}{l|lll} 
\hline
Attack
  model & \begin{tabular}[c]{@{}l@{}}Linear with \\Proda \cite{11}\end{tabular} & \begin{tabular}[c]{@{}l@{}}Linear with\\TRIM \cite{12}\end{tabular} & \begin{tabular}[c]{@{}l@{}}SVM with\\K\_LID\_SVM \cite{13}\end{tabular}  \\ 
\hline
StatP [11]          & 0.79                    & 0.7                    & 0.91   \\
Nopt [12]           & 0.815                   & 0.765                  & 0.89   \\
alfa [13]           & 0.83                    & 0.77                   & 0.91   \\
alfa \& StatP    & 0.795                   & 0.775                  & 0.885  \\
\hline
\end{tabular}
\end{table}

\begin{table}[h]
\centering
\caption{Multiple Attacks and Multiple Defense Test Results: Performance Retention Rate.}
\begin{tabular}{l|lll} 
\hline
Attack
  model & Linear with Proda & Linear with TRIM & SIDMD   \\ 
\hline
StatP          & 0.9133                  & 0.8092                 & 1       \\
Nopt           & 0.9422                  & 0.8844                 & 0.978   \\
alfa           & 0.9191                  & 0.896                  & 0.9725  \\
alfa \& StatP    & 0.9595                  & 0.8902                 & 1       \\
\hline
\end{tabular}
\end{table}

Since it is difficult to know whether the data is poisoned or not in practice, we also tested the effect of defense on the classification accuracy when the data is not poisoned.

\begin{table}[H]
\centering
\caption{Deploy defense without attack performance.}
\begin{tabular}{l|llll} 
\hline
Index                                                                & \begin{tabular}[c]{@{}l@{}}Linear with \\Proda\end{tabular} & \begin{tabular}[c]{@{}l@{}}Linear with \\TRIM\end{tabular} & \begin{tabular}[c]{@{}l@{}}SVM with \\K\_LID\_SVM\end{tabular} & SIDMD  \\ 
\hline
Accuracy                                                             & 0.785                                                       & 0.7                                                        & 0.77                                                           & 0.91   \\
\begin{tabular}[c]{@{}l@{}}Performance \\retention rate\end{tabular} & 0.9075                                                      & 0.8092                                                     & 0.9112                                                         & 1      \\
\hline
\end{tabular}
\end{table}

As can be seen from Table~7, when there is no attack, only SIDMD has no decrease in the accuracy after deployment, and the rest of the defense methods have varying degrees of influence on the classification effect of the model. Therefore, to deploy defenses such as Proda and TRIM, it is necessary to know in advance when the attacker will attack, otherwise it will have side effects on the model. However, deploying the SIDMD model does not require prior knowledge, which manifests its advantages.

\section{Related Work}

The forms of attacks against cyberspace are becoming increasingly sophisticated, and the need for cybersecurity to defend against specific attacks in specific domains is becoming more prominent. Currently, most of the work is point-to-point research on attack and defense strategies, deploying one defense against one type of attack. For example, the ability of machine learning algorithms, as tools that can assist people in making decisions for large projects, to defend against various attacks determines the future of machine learning algorithms. In addition to model attacks that are effective on specific models \cite{20}, attackers can cause poisoning attacks by manipulating or maliciously injecting anomalous data during training \cite{21}. Video recognition \cite{22}, communication networks \cite{23}, malware detection \cite{24}, edge computing \cite{25}, and domain name resolution \cite{16} have demonstrated actual poisoning attacks.

For specific models, support vector machines \cite{13}, linear regression \cite{17}, logistic \cite{14}, generative adversarial networks \cite{26}, and Bayesian networks \cite{27} have been studied concerning the attack and defense, but they only stay in point-to-point attack and defense. Sandamal et al \cite{13} proposed a weighted support vector machine, an approach that significantly reduces the classification error rate by computing K-LID for feature classification in high-dimensional transform space. Wen et al \cite{17} introduced probability estimation of clean data points in a linear regression model to reduce errors caused by poisoned datasets by optimizing the integrated model. Esmaeilpour et al \cite{26} proposed a novel GAN architecture trained with the Sobolev integral probability metric to improve the performance of the model in terms of stability and the total number of learned patterns. However, these defense techniques can only target specific attacks of specific models and are less robust in the face of multiple attacks. Therefore, for defense research in cyberspace, upgrading from point-to-point specific defense to point-to-point active defense will be the future trend.

In recent years, Moving Target Defense (MTD) \cite{28} has received more and more attention in the field of network security. In recent years, there have been many Defense Studies on false data injection attacks \cite{29}-\cite{32}, adversarial attacks \cite{33}, \cite{34}, and DDoS attacks \cite{35}, \cite{36}. MTD aims to build a dynamic, heterogeneous, and uncertain target environment in cyberspace to increase the attack difficulty of attackers. This method has pioneered the defense against external attacks by changing internal structures. However, it still does not break away from the traditional thinking of additional defense, and the problem of backdoor or trap door is difficult to defend \cite{37}. Therefore, it is necessary to introduce the redundancy theory in mimic defense to improve the overall defense performance of the system. The existing research on mimic defense mostly focuses on two aspects: heterogeneous measurement and scheduling strategy.

(1) Heterogeneity metric: For the metric of heterogeneity and similarity of an executor, it is necessary to extract the common features of endogenous security, i.e., to make a reasonable definition of the vulnerabilities of an executor so that normalized comparisons can be made between the vulnerabilities of different executors. Liu et al \cite{38} divided the executor into n parts, each of which has several optional components, and denoted the set of vulnerabilities of the executor components by the symbol, which is extended to the executor by considering the magnitude of the common-mode vulnerabilities among the components that make up the executor. Wang et al \cite{39} consider the different costs of common-mode vulnerabilities among different components of the executor and introduce weights in combining the common-mode vulnerabilities of each component of the executor, whose values are determined by the threat level of vulnerabilities in each layer. Zhang et al \cite{40} calculate the heterogeneity by comparing the higher-order heterogeneity between multiple executions is defined by comparing whether vulnerabilities exist simultaneously between the “alleles” of multiple executions in turn. \cite{38}, \cite{39}, \cite{40} all give strong assumptions on the definition of the vulnerability of the executor, and often use symbols to represent the vulnerability. In practical applications, it is often difficult for us to find out its practical significance.

(2) Scheduling strategy: most of the current studies on scheduling strategies consider the optimal defense performance at the same time, lacking the consideration of the time dimension \cite{6}, \cite{38}-\cite{41}. Wu et al. proposed a completely randomized scheduling strategy algorithm in \cite{6}, which generates the set of online executors by a completely randomized method. Liu et al. \cite{38} proposed a randomized seed scheduling strategy algorithm based on the minimum similarity between online executors (RSSSA). This algorithm randomly selects an executor as the seed executor and then schedules a set of executors with minimum similarity. Wang et al \cite{39} proposed a scheduling strategy based on Bayesian Stackelberg game theory. In the web server scenario, this algorithm can obtain the online actuator set by computing the difference between the online and offline actuator sets to maximize the security gain on the defense side. In addition, Yang et al \cite{41} proposed a policy algorithm with feedback capability, which calculates the scheduling probability of an executor based on a table of historical information. They verified the defensive performance of their algorithm by designing simulated collision experiments. For the historical confidence of executors over time and the heterogeneity among executors, Zhang et al \cite{40} proposed an optimal scheduling algorithm. They evaluated the historical confidence of the executors by setting a sliding window, which can effectively improve the operational efficiency of the algorithm and the security of the system under non-uniform distributed network attacks.

\section{Conclusion}
We perform the first systematic study on a general theoretical model of cyberspace defense. We propose a way to define vulnerabilities from the input space and a new method for calculating heterogeneity. We also design an algorithm that can dynamically update vulnerabilities based on historical attack information and have better defense performance when encountering highly dynamic attacks. To improve the lack of temporal knowledge in existing scheduling algorithms, we propose a scheduling strategy based on temporal heterogeneity, SIDMD, which significantly reduces the probability of attackers detecting common vulnerabilities in defense systems. We did numerous simulations and evaluate the performance of SIDMD on a real-world application subjected to multiple attacks simultaneously and compare SIDMD with state-of-the-art. We believe that our work will inspire future research in related fields to develop more secure anthropomorphic defense systems against various unknown attacks.

% conference papers do not normally have an appendix

% use section* for acknowledgment
%\ifCLASSOPTIONcompsoc
  % The Computer Society usually uses the plural form
 % \section*{Acknowledgments}
%\else
  % regular IEEE prefers the singular form
  %\section*{Acknowledgment}
%\fi

% We want to thank the anonymous reviewers for their valuable feedback. 

% trigger a \newpage just before the given reference
% number - used to balance the columns on the last page
% adjust value as needed - may need to be readjusted if
% the document is modified later
%\IEEEtriggeratref{8}
% The "triggered" command can be changed if desired:
%\IEEEtriggercmd{\enlargethispage{-5in}}

% references section

% can use a bibliography generated by BibTeX as a .bbl file
% BibTeX documentation can be easily obtained at:
% http://mirror.ctan.org/biblio/bibtex/contrib/doc/
% The IEEEtran BibTeX style support page is at:
% http://www.michaelshell.org/tex/ieeetran/bibtex/
%\bibliographystyle{IEEEtran}
% argument is your BibTeX string definitions and bibliography database(s)
%\bibliography{IEEEabrv,../bib/paper}
%
% <OR> manually copy in the resultant .bbl file
% set second argument of \begin to the number of references
% (used to reserve space for the reference number labels box)

% that's all folks
\end{document}